\documentclass[review]{elsarticle}

\usepackage{qcircuit}
\usepackage{float}
\usepackage{verbatim}
\usepackage{multirow}
\usepackage{booktabs}
\usepackage{longtable}
\usepackage{hyperref} 
\usepackage{mathrsfs}
\usepackage{graphicx}
\usepackage{float}
\usepackage{color}
\usepackage{braket}
\usepackage{xcolor}
\usepackage{tablefootnote}
\usepackage{graphicx}
\usepackage{amsmath}
\usepackage{cleveref}
\usepackage{amssymb}
\usepackage{amsthm}
\usepackage{amsfonts}
\usepackage{tcolorbox}
 \tcbuselibrary{skins, breakable, listings, theorems}
\usepackage{colortbl}
\newtheorem{theorem}{Theorem}
\newtheorem{corollary}{Corollary}
\newtheorem{lemma}{Lemma}

\newtheorem{definition}{Definition}

\newtheorem{conjecture}{Conjecture}

\newtheorem{Remark}{Remark}

\usepackage[margin=1.1in]{geometry}

\usepackage[linesnumbered,ruled,vlined]{algorithm2e}

\usepackage{bm}
\newcommand{\inbrace}[1]{\left \{ #1 \right \}}
\newcommand{\inparen}[1]{\left ( #1 \right )}

\newcommand{\OR}{\textsc{OR}}
\newcommand{\AND}{\textsc{AND}}
\newcommand{\Parity}{\textsc{PARITY}}
\newcommand{\Th}{\textsc{TH}}
\newcommand{\EXACT}{\textsc{EXACT}}
\newcommand{\Mod}{\textsc{MOD}}

\newcommand{\qmod}{\lceil n(1-\frac{1}{m}) \rceil}

\renewcommand{\hat}{\widehat}

\newcommand{\B}{\{0,1\}}

\let\OldLambda\lambda
\let\lambda\relax
\DeclareMathOperator{\lambda}{\OldLambda}

\begin{document}
\begin{frontmatter}
\title{On the exact quantum query complexity of $\Mod_m^n$ and $\EXACT_{k,l}^n$}

%\author[1]{Zekun Ye\authnote{1}{\indent {\it E-mail address:} yezekun@smail.nju.edu.cn}}
\author{Penghui Yao$^{a,b}$, Zekun Ye$^{a}$}
\cortext[one]{ \indent {\it E-mail address:} phyao1985@gmail.com, yezekun@smail.nju.edu.cn}

\address[1]{State Key Laboratory for Novel Software Technology, Nanjing University, Nanjing 210023, China}

\address[2]{Hefei National Laboratory, Hefei {\rm 230088}, China}

\begin{abstract}
The query model has generated considerable interest in both classical and quantum computing communities. Typically, quantum advantages are demonstrated by showcasing a quantum algorithm with a better query complexity compared to its classical counterpart. As an important complexity measure, exact quantum query complexity describes the minimum number of queries required to solve a specific problem exactly using a quantum algorithm. In this paper, we consider the exact quantum query complexity of two symmetric functions: $\Mod_m^n$, which calculates the Hamming weight of an $n$-bit string modulo $m$; $\EXACT_{k,l}^n$, which determines if the Hamming weight of an $n$-bit string is exactly $k$ or $l$. Although these two symmetric functions have received much attention, their exact quantum query complexities have not been fully characterized. Our results are as follows:
\begin{itemize}
\item 
We design an optimal quantum query algorithm to compute $\Mod_m^n$ exactly and thus provide a tight characterization of its exact quantum query complexity. Based on this algorithm, we show the exact quantum query complexity of a broad class of symmetric functions is less than their input size.
\item 
We give a tight characterization of the exact quantum query complexity of 
$\EXACT_{k,l}^n$ for some specific values of $k$ and $l$. 
\end{itemize}

\end{abstract}
\end{frontmatter}

\section{Introduction}
The quantum query model is a computational model that describes the power and limitations of quantum algorithms in solving problems in a query-based setting. It has demonstrated the powerful ability of a quantum computer to perform certain computational tasks more efficiently than a classical computer, such as Simon's algorithm \cite{simon1994power} and Shor's integer factorization algorithm \cite{shor1994algorithms}. Moreover, the quantum query model has found applications in a variety of areas, including cryptography \cite{LZ19, YZ21}, optimization \cite{GAW19, CCLW20}, and learning theory \cite{LW19, ACL+21}.

In this paper, we focus primarily on the exact quantum query complexity of symmetric functions.  The exact quantum query complexity is the minimum number of queries required to solve a specific problem exactly using quantum algorithms. As a classical counterpart, the deterministic query complexity is the minimum number of queries required to solve a specific problem with certainty using classical deterministic algorithms. A comprehensive survey on the query complexity can be found in \cite{buhrman2002complexity}. 
Symmetric functions are functions that are invariant under permutations of their inputs, which have a wide range of applications in various fields of computer science such as coding theory and cryptography. A symmetric function is partial if it is defined only on a subset of its domain, otherwise it is total.
%If the input $x$ of a symmetric function is a Boolean string, then the function value only depends on $|x|$, i.e., the number of 1's in $x$.  
\subsection{Related work}
The study of the exact quantum query complexity of partial symmetric functions has a long history. The Deutsch-Jozsa algorithm \cite{DJ92,CEMM98} demonstrated an exponential separation between exact query complexity and deterministic query complexity for the first time. Furthermore, several exact quantum algorithms \cite{Hoy00,Lon01,BHMT02} showed quadratic speedup over classical counterparts for the problem of determining whether the Hamming weight of an $n$-bit string is 0 or 1. Subsequently, Qiu and Zheng \cite{QZ2016,QZ18} determined the exact quantum query complexity and deterministic query complexity of a generalized Deutsch-Jozsa problem. He, Sun, Yang, and Yuan \cite{HSY+18} established an asymptotically optimal bound for the exact quantum query complexity of distinguishing the Hamming weight of an $n$-bit string between $k$ and $l$. Qiu and Zheng \cite{QZ2016,QZ20} studied the exact quantum query complexity for symmetric partial Boolean functions with degrees 1 or 2. In regards to the symmetric functions with large alphabet inputs, Li and Li \cite{LL22} studied the promised element distinctness problem and proposed an optimal exact quantum algorithm.

% symmetric Boolean functions are a subclass of symmetric Boolean functions that are 

%In the discussion of total symmetric Boolean functions, since the deterministic query complexity of any $n$ outlier fully symmetric Boolean function is $n$, in order to analyze the quantum advantage, we only need to discuss the quantum query complexity. 
%On the one hand, the best-known quantum algorithms so far require at least $\frac{n}{2}$ lookups for any non-nominal completely symmetric Boolean function. On the other hand, combining the work of von zur Gathen and Roche\cite{vzGR97} with the polynomial approach yields a quantum lower bound \cite{BBC+01}, the following conclusions can be drawn: the exact quantum algorithm to calculate any $n$ non-value complete symmetric Boolean function needs at least $\frac{n}{2}-O(n^{0.548}) $ queries. Furthermore, Ambainis et al. \cite{AGZ15} demonstrated that, Except for $\text{AND}_n$ and its isomorphic functions, the exact quantum query complexity of all other $n$-bit complete Boolean symmetric functions is less than $n$. The above results are an overall description of the exact quantum query complexity of complete Boolean symmetric functions, while 

The exact quantum query complexity of total symmetric functions has also been studied extensively. On the one hand, the best-known exact quantum algorithm for any $n$-bit non-constant symmetric Boolean function requires at least $n/2$ queries. On the other hand, combining the lower bound on the degree of symmetric Boolean functions \cite{vzGR97}, the best-known result about the difference between consecutive primes \cite{BHP01}
with polynomial methods \cite{BBC+01}, it leads to the following conclusion: any exact quantum algorithm for computing any $n$-bit non-constant symmetric Boolean function requires at least $n/2-O(n^{0.525})$ queries. Moreover, Montanaro, Jozsa and Mitchison \cite{MJM15} indicated the exact quantum query complexity of all symmetric Boolean functions on up to 6 bits by numerical results. %Chen, Ye and Li \cite{CYL20} proved $\Parity_1$ and $\Parity_2$ are the only two total Boolean functions, up to isomorphism, that can be computed by one-query quantum algorithm exactly. 
Ambainis, Gruska and Zheng \cite{AGZ15} showed $\AND_n$ is the only $n$-bit Boolean function, up to isomorphism, that requires $n$ quantum queries to compute exactly. 
While the deterministic query complexity of all non-constant total symmetric functions is $n$ \cite{Aaronson03,MJM15,AGZ15},
there are only a few total symmetric functions whose exact quantum query complexity is fully characterized, which are summarized as Table \ref{table_exact} (up to isomorphic). Note that the functions $\neg\OR_n$ and $\AND_n$  are special cases of $\EXACT_k^n$ when $k = 0$ and $n$.

\begin{table}
\caption{The exact quantum query complexity of several symmetric functions}
\label{table_exact}
\begin{center}

\begin{tabular}{|l|p{7cm}|l|}
\hline
Functions & Definition & Exact Quantum Query Complexity \\
\hline
$\Parity_n$ & $\Parity_n(x) = |x| \bmod 2$ & $\lceil n/2 \rceil$ \cite{CEMM98,FGGS1998,BBC+01}\\
\hline
$\EXACT_k^n$ & $
\EXACT_k^n(x) = 
\begin{cases}
1, &\text{if }|x| = k, \\
0, &\text{otherwise}.
\end{cases}
$ & $\max\inbrace{k,n-k}$ \cite{AIS13}\\
\hline
$\Th_k^n$ & $
\Th_k^n(x) = 
\begin{cases}
1, &\text{if }|x| \ge k, \\
0, &\text{otherwise}.
\end{cases}
$ & $\max\inbrace{k,n-k+1}$ \cite{AIS13}\\
\hline
\end{tabular}
\end{center}
\end{table}

There are some total symmetric functions that have been studied, but the exact quantum query complexity is not fully characterized, including $\Mod_m^n$ and $\EXACT_{k,l}^n$. 
Specifically, $\Mod_m^n$ aims to compute the Hamming weight of an $n$-bit string modulo $m$, which is a generalization of $\Parity_n$. Recently, Cornelissen, Mande, Ozols and de Wolf \cite{CMO+21} showed that when the prime factor of $m$ is only 2 or 3, the exact quantum query complexity of $\Mod_m^n$ is $\qmod$. Moreover, they proved the exact quantum query complexity of $\Mod_m^n$ is at least $\qmod$ for any $1 < m \le n$. Then they conjectured the lower bound is tight as Conjecture \ref{con:mod}. Afterward, using variational learning algorithms, Wu, Hou, Zhang, Li and Zeng \cite{WHZ+22} suggested that when $m = n = 5$, there exists a quantum algorithm using 4 queries to compute $\Mod_m^n$, which is consistent with Conjecture \ref{con:mod}.
\begin{conjecture}[\cite{CMO+21}]\label{con:mod}
For $1 < m \le n$, the exact quantum query complexity of $\Mod_m^n$ is $\qmod$.
\end{conjecture}
For the function $\EXACT_{k,l}^n$ $(k < l)$, which aims to determine whether $|x| \in \inbrace{k,l}$, Ambainis, Iraids and Nagaj \cite{AIN17} gave the best-known result: the exact quantum query complexity of $\EXACT_{k,l}^n$ falls within a range of $\max\inbrace{{n-k,l}-1}$ to $\max\inbrace{{n-k,l}+1}$. Moreover, they showed if $l-k \in \{2,3\}$, the lower bound is tight and proposed Conjecture \ref{con:kl}.
For $\EXACT_{k,l}^n$, Wu et al. \cite{WHZ+22} also gave the numerical result about some instances of small sizes. For the case $l-k = 2$, $n$ is even and $k,l$ is symmetrically distributed around $n/2$, the numerical result is consistent with Conjecture \ref{con:kl}.
\begin{conjecture}[\cite{AIN17}]\label{con:kl}
If $l-k \ge 2$, the exact quantum query complexity of $\EXACT_{k,l}^n$ is $\max\{n-k,l\}-1$.
\end{conjecture}
\subsection{Our contribution}
In this paper, we consider the above two conjectures and study the exact quantum query complexity of $\Mod_m^n$ and $\EXACT_{k,l}^n$. Our motivation is as follows:
\begin{itemize}
\item The exact quantum query complexity of $\Mod_m^n$ and $\EXACT_{k,l}^n$ are not fully characterized. Thus, we aim to improve the best-known result of these two functions.
\item In the quantum model, we say a function is \emph{evasive} if its exact quantum query complexity equals its input size. $\Mod_m^n$ is a key function to analyze the quantum evasiveness of the symmetric functions with large alphabet output. By studying the exact quantum query complexity of $\Mod_m^n$, we can better understand the quantum evasiveness of a broad class of symmetric functions.
\item At present, there are quite a few exact quantum algorithm design techniques. It is interesting to obtain more exact quantum algorithm design paradigms.
\end{itemize}
Our contribution is as follows:
i) We propose an optimal quantum query algorithm to compute $\Mod_m^n$ exactly and thus prove Conjecture \ref{con:mod}. % with $\qmod$ queries. 
%Thus, we prove the conjecture proposed by Cornelissen et al. \cite{CMO+21}. 
Compared to the algorithm proposed in \cite{CMO+21}, our algorithm is more natural and suitable for any $1 < m \le n$. As a corollary, we prove a wide range of symmetric functions is not evasive in the quantum model based on the above algorithm.
ii) We prove Conjecture \ref{con:kl} for the case $k=0$ and $k=1,l=n-1$. Thus, we %partially prove the conjecture proposed by Ambainis et al. \cite{AIN17} and 
give a tighter characterization to the exact quantum query complexity of $\EXACT_{k,l}^n$.

 \subsection{Organization}
The remainder of the paper is organized as follows. In Section \ref{sec:pre}, we review some definitions and notations used in this paper. In Section \ref{sec:mod}, we give an optimal exact quantum query algorithm to compute $\Mod_m^n$ and analyze the quantum evasiveness of a broad class of symmetric functions. In Section \ref{sec:kl}, we discuss the exact quantum query complexity of $\EXACT_{k,l}^n$. Finally, a conclusion is made in Section \ref{sec:con}.

\section{Preliminary}\label{sec:pre}
This section first gives some formal definitions of the quantum query model. For convenience, for an $n$-bit Boolean string $x \in \B^n$, we let $x= x_0\cdots x_{n-1}$. For positive integer $n$, let $[n] = \inbrace{1,\dots,n}$ and $\mathbb{Z}_n = \inbrace{0,1,...,n-1}$.
\begin{definition}[POVM \cite{NC15}]
A set of operators $\inbrace{E_j}$ is a POVM (Positive Operator-Valued Measure) if each operator $E_j$ is positive and $\sum_j E_j = I$. If a measurement  described by $\inbrace{E_j}$ is performed upon a quantum state $\ket{\psi}$, then the probability of obtaining outcome $j$ is given by $p(j) = \langle \psi | E_j | \psi \rangle$.
\end{definition}
\begin{definition}[Quantum query algorithms]
\label{def:qqa}
A quantum query algorithm $\mathcal{A}$ consists of an initial  state $\ket{\psi_0}$, a unitary operator sequence $U_TO_xU_{T-1}O_x\cdots O_xU_0$ and a POVM $\{E_j\}$, where $U_i$'s are fixed unitary operators, and $O_x$ is a quantum query oracle dependent on $x \in \{0,1\}^n$, which is defined as $O_x \ket{i}\ket{b} = \ket {i}\ket{b + x_i}$, where $i \in \{0,...,n-1\}$, $b \in \mathbb{Z}_n$. The algorithm process is as follows:
\begin{itemize}
    \item Prepare the initial state $\ket{\psi_0}$;
    \item Perform unitary operations $U_0,O_x,...,O_x,U_T$ sequentially on $\ket{\Psi_0}$ to obtain the quantum state $\ket{\Psi_x} = U_TO_xU_{T-1}O_x\cdots O_xU_0 \ket{\Psi_0}$;
\item Perform the measurement described by $\{E_j\}$ upon the quantum state $\ket{\Psi_x}$, %. Then the probability of outcome $m$ is given by $p(m) = \langle \Psi_x | E_j | \Psi_x \rangle$. We 
use the measurement result as the output $\mathcal{A}(x)$ of the algorithm.
\end{itemize}
The query complexity of the algorithm is defined as the number of query oracle $O_x$ used in the algorithm.
\end{definition}
\begin{Remark}
	The quantum query oracle $O_x$ is the linear extension of the reversible  mapping $(i,b) \rightarrow (i,b+x_i)$, which can be performed efficiently given the ability to compute $i\rightarrow x_i$ efficiently \cite{Childs2021Lecture}.
\end{Remark}

As mentioned in \cite{AIS13}, a quantum algorithm can also be described as a recursive algorithm with the following structure:
First, perform unitary operation $U_1O_xU_0$ and measure; second, depending on the measurement result, call a smaller instance of the algorithm. Such a recursive algorithm can be transformed into a quantum query algorithm described as Definition \ref{def:qqa} with the same query complexity.
\begin{definition}[Exact quantum algorithms]
Given function $f:\B^n \rightarrow X$, where $X$ is a finite set, if a quantum algorithm $\mathcal{A}$ satisfies $\mathcal{A}(x) = f(x)$ for any $x \in \B^n$, then $\mathcal{A}$ is an exact quantum algorithm to compute $f$.
\end{definition}

\begin{definition}[Exact quantum query complexity]
For function $f:\B^n \rightarrow X$, where $X$ is a finite set. The exact quantum query complexity of $f$, $Q_E(f)$, is the minimal number of queries an exact quantum algorithm requires to compute $f$. 
%on the worst input.
%the query complexity of the optimal exact quantum algorithm to compute $f$.
\end{definition}

\begin{definition}[Univariate version of symmetric functions]
For a symmetric function $f: \B^n \rightarrow X$, where $X$ is a finite set, we define $F: \inbrace{0,1,...,n} \rightarrow X$ as $F(x) = f(|x|)$ for any $x \in \B^n$, where $|x|$ is the Hamming weight of $x$, i.e., the number of $1$'s in $x$. 
\end{definition}

\begin{definition}[Majority index]
Suppose $x \in \B^n$ and $|x| \neq \frac{n}{2}$, we say $i$ is a majority index of $x$ if i) $|x| >  \frac{n}{2}$ and $x_i = 1$, or ii) $|x| < \frac{n}{2}$ and $x_i = 0$.
\end{definition}
Next, we give some notations used in this paper. Let $\theta = 2\pi/n$,
\begin{equation}
	F_n = \frac{1}{\sqrt{n}}\sum_{j,k\in\mathbb{Z}_n}e^{ijk\theta}\ket{j}\bra{k},
\end{equation}
and
\begin{equation}\label{eq:hatO}
\hat{O}_x = (I \otimes F_n) O_x (I \otimes F^{\dag}_n). 
\end{equation}
By the phase kickback trick \cite{Childs2021Lecture}, we have
\[
\hat{O}_x\ket{j,b} = e^{ib\theta x_j}\ket{j,b},
\]
where $j,b \in \mathbb{Z}_n$.
%For any angle $\theta$, let 
%$$
%\ket{\theta}^{\circ} =  \cos \theta \ket{0} +\sin \theta \ket{1}.\\
%$$
For any $i,a\in \mathbb{Z}_n$, the permutation operation $U_a$ is defined as 
\begin{equation}\label{eq:Ua}
	U_a\ket{i} = \ket{i+a}.
\end{equation}
Let 
\begin{equation}\label{eq:oracle}
O_{x,a} = (U_{a} \otimes I)\hat{O}_x (U_{-a} \otimes I).
\end{equation}
Then 
$$
O_{x,a}\ket{j,b} = e^{ib\theta x_{j-a}}\ket{j,b},
$$
where $j, b \in \mathbb{Z}_n$, and the subtraction operator is with modulo $n$.
%\begin{cases}
%\ket{i}\ket{0}^{\circ}, & \text{ if }x_{i-a} = 0, \\ 
%\ket{i}\ket{\theta}^{\circ}, & %\text{ if }x_{i-a} = 1, 
%\end{cases}
%For $i\in \inbrace{0,1}$ and any angle $\eta$, the unitary operator $\text{CR}_{\theta}$ is defined as
%\begin{equation}\label{eq:CR}
%\text{CR}_{\theta}\ket{\eta}^{\circ}\ket{i} = \ket{\eta+i\theta}^{\circ}\ket{i}. 
%\end{equation}
\begin{comment}
\begin{equation}\label{eq:CR}
\begin{cases}
\text{CR}_{\theta} \ket{\eta}^{\circ} \ket{0}&= \ket{\eta}^{\circ}\ket{0},  \\
\text{CR}_{\theta}\ket{\eta}^{\circ}\ket{1} &= \ket{\eta+\theta}^{\circ}\ket{1}.
\end{cases}
\end{equation}
\end{comment}
%\begin{cases}
%\ket{i}\ket{\eta}^{\circ}, & \text{ if }x_{i-a} = 0, \\ 
%\ket{i}\ket{\theta+\eta}^{\circ}, & \text{ if }x_{i-a} = 1, 
%\end{cases}
%$$
%\theta_{i-a} =
%\begin{cases} 
%0 &\text{ if }x_{i-a} = 0, \\
%\theta & \text{ if }x_{i-a} = 1, \\
%\end{cases}
%$$
%and 
%$$
%\SUM\ket{0}\ket{\eta_1}\ket{\eta_2}\cdots\ket{\eta_n} = \ket{\sum_{i=1}^n\eta_i}\ket{\eta_1}\ket{\eta_2}\cdots\ket{\eta_n}.
%$$
%We select the unitary operation $\SUM$ such that
%\begin{equation}
%\label{eq:sum}
%\SUM\ket{0}^{\circ}\ket{\eta_1}^{\circ}\ket{\eta_2}^{\circ}\cdots\ket{\eta_{n-1}}^{\circ} = |\sum_{i=1}^{n-1}{{\eta_i}\rangle}^{\circ}\ket{\eta_1}^{\circ}\ket{\eta_2}^{\circ}\cdots\ket{\eta_{n-1}}^{\circ}
%\end{equation}
%for any angle $\eta_1,...,\eta_{n-1}$.
For $k \in \{0,...,n-1\}$, let 
\[
\ket{\phi_{k}} = \frac{1}{\sqrt{n}}\sum_{j = 0}^{n-1}e^{ijk\theta} \ket{j}.
\]
%where $\theta = \frac{2\pi}{n}$. Then $\langle \phi_j | \phi_j^* \rangle = 0$.
%, 
Then for any $k,l \in \{0,...,n-1\}$ and $k\neq l$, %we have $\sin \frac{n(j-k)\theta}{2} = 0$. Thus, by Fact \ref{fact:sum_tri}, we have
\[
\langle \phi_k | \phi_l \rangle = \frac{1}{n}\sum_{j=0}^{n-1}e^{ij(k-l)\theta}= 0.
\]
As a result, $\inbrace{\ket{\phi_0},...,\ket{\phi_{n-1}}}$ is an orthonormal basis. 
%Let $I_{d}$ be the identity operator in a $d$-dimension Hilbert space. 
Let 
\begin{equation}\label{eq:Pj}
	P_j =
\ket{\phi_j}\bra{\phi_j} \text{ for $j \in \inbrace{0,...,n-1}$.}
\end{equation}
Then for any $j \in \inbrace{0,...,n-1}$, $P_j$ is a projection operator and $\sum_{j=0}^{n-1} P_j = I$. Thus, $\inbrace{P_j}$ is a POVM.

\section{Computing the Hamming weight modulo $m$}\label{sec:mod}
In this section, we present an optimal algorithm to compute $\Mod_m^n$, which is defined as $\Mod_m^n(x) = |x| \bmod m$ for any $x \in \B^n$. First, we give Algorithm \ref{algorithm1} to compute $\Mod_n^n$, where $\theta = 2\pi/n$. We verify the correctness of Algorithm \ref{algorithm1} as follows. For any $x\in \{0,1\}^n$, let $S_x = \{j:x_j = 1\}$. 
If $j-a = l$ for some $l \in S_x$, then $a\theta x_{j-a} = (j-l)\theta$; otherwise, $a\theta x_{j-a} = 0$. Thus, we have 
\begin{equation}\label{eqSx1}
	\sum_{a=1}^{n-1}a\theta x_{j-a} = \sum_{a=0}^{n-1}a\theta x_{j-a} = \sum_{l \in S_x}(j-l)\theta.
\end{equation}
If $|x| \bmod m = k$, then 
\begin{equation}\label{eqSx2}
	\sum_{l \in S_x}(j-l)\theta = jk\theta+\eta_x,
\end{equation}
where $\eta_x$ only depends on $x$ and $\eta_x = -\sum_{l \in S_x} l\theta$.
By \Cref{eqSx1,eqSx2}, for the final state $\ket{\psi_x}$ of Algorithm \ref{algorithm1}, we have
$$
\begin{aligned}
\ket{\psi_x} &= \frac{1}{\sqrt{n}}\sum_{j = 0}^{n-1} \exp\inparen{{i\sum_{a=1}^{n-1}a\theta x_{j-a}}}\ket{j}\ket{n-1}\\%\frac{1}{\sqrt{n}}\sum_{i = 0}^{n-1} \ket{i}	\ket{ik\theta+\eta_x}^{\circ}\ket{x_{i-1}}\ket{x_{i-2}}\cdots \ket{x_{i-n}} \\
&= e^{i\eta_x}\frac{1}{\sqrt{n}}\sum_{j = 0}^{n-1} e^{ijk\theta} \ket{j}\ket{n-1}\\
&= e^{i\eta_x}\ket{\phi_k}\ket{n-1}.
\end{aligned}
$$
%Therefore, the probability of obtaining the measurement result $k$ is $1$.
Thus, the algorithm will output $k$ with the probability 1 after performing the measurement described by $\{P_j\}$ as \Cref{eq:Pj}, i.e., Algorithm \ref{algorithm1} always outputs the correct result.
\begin{algorithm}
			\caption{Compute $\Mod_n^n$.}
			\label{algorithm1}
			\KwIn{$x\in\B^n$;}
			\KwOut{$|x| \bmod n$.}
			\begin{enumerate}
				\item Prepare initial state $\ket{\psi_0} = \frac{1}{\sqrt{n}}\sum_{j = 0}^{n-1} \ket{j}\ket{0}$. % which consists of an $n$-dimensional qudit and $n$ ancillary qubits.
				\item For $a = 1$ to $n-1$, perform the following two operations sequentially: i) perform $U_1$ as \Cref{eq:Ua}  in the second register; ii) perform $O_{x,a}$ as \Cref{eq:oracle} in the whole registers. Finally, the final state is 
				$$
				\ket{\psi_x} = \frac{1}{\sqrt{n}}\sum_{j = 0}^{n-1} \exp\inparen{i\sum_{a=1}^{n-1}a\theta x_{j-a}}\ket{j}\ket{n-1}.
				$$
				%\item For $a=1$ to $n-1$, we perform $\text{CR}_{a\theta}$ given by Equation (\ref{eq:CR}) in the first ancillary qubit the $(a+1)$-th ancillary qubit
			%	sequentially to obtain the final state 
				%$$
				%\ket{\psi_x} = \frac{1}{\sqrt{n}}\sum_{i = 0}^{n-1} \ket{i}	|\sum_{a=1}^{n-1}{a\theta x_{i-a}\rangle}^{\circ}\ket{x_{i-1}}\ket{x_{i-2}}\cdots \ket{x_{i-(n-1)}}.
				%$$
				\item Perform the measurement described by $\inbrace{P_j}$ defined in Equation (\ref{eq:Pj}) upon the first register of quantum state $\ket{\psi_x}$, and then output the measurement result. 
			\end{enumerate}
	\end{algorithm}
Since we have $O_{x,a} = (U_{a} \otimes F_{n})O_x (U_{-a} \otimes F^{\dag}_{n})$ as  \Cref{eq:hatO,eq:oracle}, the number of query oracle $O_x$ used in Algorithm \ref{algorithm1} is $n-1$. 

%(I \otimes F_{\mathbb{Z}_n}) O_x (I \otimes F_{\mathbb{Z}_n})

Next, for $1 < m < n$, let $c = \lfloor\frac{n}{m}\rfloor$ and $n = cm+q$. Then $0 \le q < m$. 
We give Algorithm \ref{algorithm2} to compute $|x| \bmod m$. The algorithm procedure is as follows. i) If $q = 0$, we partition $x$ into $m$-bit substrings $x^{(0)},\dots,x^{(c-1)}$. For any $0 \le i \le c-1$, we compute $b_i = |x^{(i)}| \bmod m$ by Algorithm \ref{algorithm1}. Finally, we output $(\sum_{i=0}^{c-1} b_i) \bmod m$. ii) If $q \neq 0$, we partition $x$ into $c$ $m$-bit substrings $\inbrace{x^{(0)},\dots,x^{(c-1)}}$ and one $q$-bit substring $x^{(c)}$. For $0 \le i \le c-1$, we compute $b_i = |x^{(i)}| \bmod m$ by Algorithm \ref{algorithm1}. Then we query all the elements in $x^{(c)}$ and compute $b_c = |x^{(c)}| \bmod m$. Finally, we output $\inparen{\sum_{i=0}^c b_i} \bmod m$. We verify the correctness of Algorithm \ref{algorithm2}. If $q = 0$, then 
$$
\begin{aligned}
|x| \bmod m &= \inparen{\sum_{i=0}^{c-1} |x^{(i)}|} \bmod m \\
&= \inparen{\sum_{i=0}^{c-1} \inparen{|x^{(i)}| \bmod m}} \bmod m \\
&= \inparen{\sum_{i=0}^{c-1} b_i} \bmod m.
\end{aligned}
$$
If $q > 0$, then we have $|x| \bmod m = \inparen{\sum_{i=0}^{c} b_i} \bmod m$ similarly. Thus, Algorithm \ref{algorithm2} always gives the correct output. Moreover, the number of queries in Algorithm is $c(m-1)+q = n-c = \lceil n(1-\frac{1}{m})\rceil$.
	
\begin{algorithm}
\caption{Compute $\Mod_m^n$.}
\label{algorithm2}
			\KwIn{$x\in\B^n$, integers $m,c,q$ such that $1 < m < n$, $c = \lceil \frac{n}{m} \rceil$ and $n = cm+q$;}
			\KwOut{$|x| \bmod m$.}
%\begin{algorithmic}[1]
\For{$i = 0 \to c-1$}{
 Let $x^{(i)} = x_{im}\cdots x_{(i+1)m-1}$\;
 Compute $b_i = |x^{(i)}| \bmod m$ by Algorithm \ref{algorithm1}\;
}
\eIf{$q=0$}{\Return $(\sum_{i=0}^{c-1} b_i) \bmod m$;}{
Let $x^{(c)} = x_{cm}\cdots x_{n-1}$\;
 Query all the elements in $x^{(c)}$ and let $b_c = |x^{(c)}|$\;
 \Return $\inparen{\sum_{i=0}^c b_i} \bmod m$;
}
%\end{algorithmic}
\end{algorithm}

The above results %show that for any $1 < m \le n$, there exists an exact quantum query algorithm to compute $\Mod_m^n$ with $\qmod$ queries, which 
implies the following theorem: %\ref{th:mod}.
\begin{theorem}\label{th:mod}
For $1 < m \le n$, there exists an exact quantum query algorithm to compute $\Mod_m^n$ using $\qmod$ queries.
\end{theorem}
Since Cornelissen et al. \cite{CMO+21} showed that any quantum algorithm needs at least $\lceil n(1-\frac{1}{m})\rceil$ queries to compute $|x| \bmod m$ exactly\footnote{While  \cite{CMO+21} uses a slightly different quantum query oracle, it is not hard to check their proof of lower bound also works for our oracle $O_x$.}, our algorithm is optimal and Conjecture \ref{con:mod} is proved. As an implication of Theorem \ref{th:mod}, we show the following corollary:
%we can compute $|x| \bmod n$ using $n-1$ queries for $x\in\B^n$. Furthermore, 
%we show there exists an exact quantum algorithm using less than $n$ quantum queries to compute any symmetric function $f:\B^n \rightarrow X$ such that $F(0) = F(k)$ and $F(n-k) = F(n)$ for some $k \in [n]$, where $F$ is the univariate version of $f$. 
\begin{corollary}\label{th:lessn}
For any symmetric functions $f:\B^n \rightarrow X$, where $X$ is a finite set, let $F(|x|) = f(x)$ for any $x \in \B^n$. If there exists $k \in [n]$ such that $F(0) = F(k)$ and $F(n-k) = F(n)$, then the exact quantum query complexity of $f$ is less than $n$. Moreover, the upper bound is tight, i.e., there exists a symmetric function $f$ satisfying the above conditions whose exact quantum query complexity is $n-1$.
\end{corollary}
\begin{proof}
If $k = n$, we compute $a = |x| \bmod n$ using $n-1$ quantum queries by Algorithm \ref{algorithm1} and then $f(x) = F(a)$. If $k \in \inbrace{1,...,n-1}$, we give Algorithm \ref{al:lessn} to compute $f$. The algorithm procedure is as follows. First, we partition $x$ into two substrings $x' \in \B^{n-k}$ and $x'' \in \B^k$. Then we compute $a = |x'| \bmod (n-k)$ and $b = |x''| \bmod k$ by Algorithm \ref{algorithm1}. Then we discuss the following cases:
\begin{itemize}
    \item If $a\neq 0, b \neq 0$, then $|x| = a+b$.
    \item If $a \neq 0,b=0$, then we query $x_{n-k}$ to determine $|x''| = 0$ or $k$, and thus determine $|x|$.
    \item If $a=0,b \neq 0$, the case is similar to the above case.
    \item  if $a=0,b=0$, then we query $x_{0}$. If $x_0 = 0$, then $|x| = 0$ or $k$; if $x_0 = 1$, then $|x| = n-k$ or $n$.
\end{itemize}

\begin{algorithm}
\caption{Compute $f$} 
\label{al:lessn}
\SetKwInOut{Input}{Input}
\SetKwInOut{Output}{Output}
\KwIn{$x \in \B^n$, a symmetric function $f:\B^n \rightarrow \B$ such that $F(0) = F(k)$ and $F(n-k) = F(n)$ for some $k \in [n]$, where $F$ is the univariate version of $f$;}
\Output{$f(x)$.}

Let $x' = x_0\cdots x_{n-k-1},x'' = x_{n-k}\cdots x_{n-1}$;

Compute $a = |x'| \bmod (n-k)$ and $b = |x''| \bmod k$ using Algorithm \ref{algorithm1};

\Switch{$(a, b)$}{
\Case{$(a \neq 0, b \neq 0)$}{
$f(x) = F(a+b)$;
}
\Case{$(a \neq 0, b = 0)$} {
Query $x_{n-k}$\;
\lIf{$x_{n-k} = 0$}{$f(x) = F(a)$}\lElse{$f(x) = F(a + k)$}
}
\Case{$(a = 0, b \neq 0)$} {
Query $x_0$\;
\lIf{$x_{0} = 0$}{$f(x) = F(b)$}\lElse{$f(x) = F(n-k+b)$}
}
\Case{$(a = b = 0)$} {
Query $x_0$\;
\lIf{$x_{0} = 0$}{$f(x) = F(0)$}\lElse{$f(x) = F(n-k)$}
}
}
\end{algorithm}

The correctness of Algorithm \ref{al:lessn} is not hard to verify. Moreover, the number of queries of the algorithm is at most $(n-k-1)+(k-1)+1 = n-1$, i.e., $Q_E(f) \le n-1$. Since $Q_E(\Mod_n^n) = n-1$, the above bound is tight. %sThus we prove Corollary \ref{th:lessn}.

\end{proof}

Furthermore, Ambainis et al. \cite{AGZ15} proved that the exact quantum query complexity of a total symmetric Boolean function $f:\{0,1\}^n \rightarrow \{0,1\}$ is $n$ if and only if $f$ is isomorphic to $\AND_{n}$ function. Correspondingly, we conjecture there exists a generalized characterization to all total symmetric functions $f:\{0,1\}^n \rightarrow X$, where $X$ is a finite set. %Let $F : \inbrace{0,...,n} \rightarrow X$ be defined as $F(|x|) = f(x)$ for any $x \in \B^n$ which is the univariate function with respect to $F$. 
Thus we give the following conjecture:

\begin{conjecture}\label{con1}
Given a total symmetric function $f:\{0,1\}^n \rightarrow X$, where $X$ is a finite set. Let $F$ be the univariate version of $f$. Then the exact quantum query complexity of $f$ is $n$ if and only if one of the following conditions satisfies:
\begin{itemize}
    \item [i)] $F(0) \neq F(i)$ for any $i \in [n]$; %\label{condition1}
     \item [ii)] $F(n) \neq F(i)$ for any $i \in \inbrace{0,\dots,n-1}$. %\label{condition2}.
\end{itemize}
\end{conjecture}

Suppose a function $f$ satisfies item i). Then for any $x$, if there exists an algorithm to compute $f(x)$, then the algorithm also can compute $\AND_n(x)$, and thus  $Q_E(f) \ge Q_E(\AND_n)$. Similarly, if $f$ satisfies item ii), then $Q_E(f) \ge Q_E(\OR_n)$.
Thus, $Q_E(f) = n$. As a result, to solve the above conjecture, we only need to solve the following question: if there exist $i \in [n],j\in\inbrace{0,...,n-1}$ such that $F(0) = F(i)$ and $F(n) = F(j)$, whether the exact quantum query complexity of $f$ is less than $n$? By Corollary \ref{th:lessn}, we have already proven that if $i+j = n$, then the exact quantum query complexity of $f$ is less than $n$. Thus, we propose the following conjecture:
\begin{conjecture}\label{con2}
If a total symmetric function $f:\B^n \rightarrow X$ satisfies $F(0) = F(i)$, $F(j) = F(n)$ for some $i \in [n],j \in \inbrace {0,...,n-1}$ such that $i+j \neq n$, then $Q_E(f) < n$, where $F$ is the univariate version of $f$.
\end{conjecture}
If Conjecture \ref{con2} is proved, then Conjecture \ref{con1} is also correct.

\section{Exact Quantum Query Complexity of $\EXACT_{k,l}^n$}\label{sec:kl}
In this section, we consider the exact quantum query complexity of $\EXACT_{k,l}^n$ for $l-k \ge 2$. The $n$-bit Boolean function $\EXACT_{k,l}^n$ is defined as follows:
$$
\EXACT_{k,l}^n(x) = 
\begin{cases}
1, &\text{if }|x| \in \{k,l\}, \\
0, &\text{otherwise}.
\end{cases}
$$
In the following context, we need to use the $n$-bit Boolean function $\EXACT_k^n$, defined as
$$
\EXACT_{k}^n(x) = 
\begin{cases}
1, &\text{if }|x| = k, \\
0, &\text{otherwise}.
\end{cases}
$$
First, we consider the case $k = 0$. We give the following lemma:
\begin{lemma}\label{lemma:klzero}
For $x \in \B^n$ and $2 \le l \le n$, there exists a quantum algorithm to compute $\EXACT_{0,l}^n(x)$ with $n-1$ queries.
\end{lemma}

\begin{proof}
If $l < n$, we provide the algorithm as follows. For $i=0$ to $n-l-1$, we query $x_i$ until $x_i = 1$ for some $i$. Then we consider the following two cases:
i) If we find the smallest integer $i \in [0,n-l-1]$ such that $x_i = 1$, let 
$x' = x_{i+1}\cdots x_{n-1}$. Then $\EXACT_{0,l}^n(x) = \EXACT_{l-1}^{n-i-1}(x')$. Since $\EXACT_{l-1}^{n-i-1}(x')$ can be computed by $\max\inbrace{n-i-l,l-1}$ quantum queries \cite{AIS13}, the total number of queries is
$$
\begin{aligned}
(i+1)+ \max\inbrace{n-i-l,l-1} &\le  \max\inbrace{n-l,i+l} \\
&\le \max\inbrace{n-l,n-1} \\
&= n-1.
\end{aligned}
$$
ii) If we find $x_i = 0$ for any $0 \le i \le n-l-1$, let $x' = x_{n-l} \cdots x_{n-1}$ and compute $|x'| \mod l$ using Algorithm \ref{algorithm1}. If $|x'| \mod l =0$, then $|x| = 0$ or $l$, and thus $\EXACT_{0,l}^n(x) = 1$; otherwise, $\EXACT_{0,l}^n(x) = 0$. The total number of queries is $n-l+l-1 = n-1$.

If $l = n$, we compute $|x| \bmod n$ using $n-1$ quantum queries by Algorithm \ref{algorithm1}, and then $|x|\in \inbrace{0,l}$ if and only if $|x|\bmod n = 0$. 
\end{proof}
Second, we consider the case $k=1$ and $l =n-1$. We give the following lemma. \begin{lemma}\label{lemma:klone}
For $x\in\B^n$ and $n \ge 4$, there exists a quantum algorithm to compute $\EXACT_{1,n-1}^n(x)$ with $n-2$ queries.
\end{lemma}

%\begin{algorithm}[H]
%\label{al:exact_one_n}
%\caption{Compute $\EXACT_{1,n-1}^n(x)$.}
    %If $n$ is even, then we compute $x_0 \oplus x_1, x_2 \oplus x_3, ..., x_{n-2}\oplus x_{n-1}$ by $\frac{n}{2}$ quantum queries. If exactly one pair $(x_{2i},x_{2i+1})$ are equal, we let $x' = x_0x_2\cdots x_{2i-2}x_{2i+2}\cdotsx_{n-2}$. Then $|x'| \in \inbrace{0,\frac{n}{2}-1}$ if and only if $|x| \in \inbrace{1,n-1}$. Then we determine whether $|x'| \in \inbrace{0,\frac{n}{2}-1}$ by Algorithm \ref{algorithm1} using $\frac{n}{2}-2$ queries. Thus, the number of queries is $\frac{n}{2}+\frac{n}{2}-2 = n-2$ queries.
    \begin{proof}
    We give a recursive algorithm as follows. The goal of the algorithm is to determine whether $|x| \in \inbrace{1,n-1}$. If $|x| \in \inbrace{1,n-1}$, the algorithm finds at least a majority index of $x$. %In the algorithm, a subroutine is to compute $x_i \oplus x_j$ by 1 quantum query, which can be implemented by Deutsch's Algorithm or our Algorithm \ref{algorithm1}.
    \begin{itemize}
        \item If $n = 4$, we compute $x_0 \oplus x_1$ and $x_2 \oplus x_3$ using 2 quantum queries by Algorithm \ref{algorithm1}. i) If $x_0 \oplus x_1 = 0, x_2 \oplus x_3 = 1$, then $|x| \in \inbrace{1,3}$ and $\inbrace{0,1}$ are majority indices of $x$; ii) if $x_0 \oplus x_1 = 1, x_2 \oplus x_3 = 0$, then $|x| \in \inbrace{1,3}$ and $\inbrace{2,3}$ are majority indices of $x$ similarly; iii) if $x_0 \oplus x_1 =  x_2 \oplus x_3$, then $|x| \in \inbrace{0,2,4}$.
    
    \item If $n= 5$, then there exists an algorithm to determine whether $|x| \in \inbrace{1,4}$ using 3 quantum queries \cite{AIN17}. It is worth noting if $|x| \in \inbrace{1,4}$, the algorithm will find some $i,j$ such that $x_i \neq x_j$. Thus, all the indices except $i,j$ are the majority indices of $x$.
    
    \item If $n > 5$, let $x' = x_2\cdots x_{n-1}$. We compute $x_0 \oplus x_1$ using 1 quantum query first. i) If $x_0 \neq x_1$, 
    %then $\EXACT_{1,n-1}^n(x) = \EXACT_{0,n-2}^{n-2}(x')$. Thus, 
    we compute $|x'| \bmod n-2$ using Algorithm \ref{algorithm1} to determine whether $|x'| \in \inbrace{0,n-2}$. If $|x'| \notin \inbrace{0,n-2}$, then $|x| \notin \inbrace{1,n-1}$; if $|x'| \in \inbrace{0,n-2}$, then $|x| \in \inbrace{1,n-1}$ and $\inbrace{2,...,n-1}$ are majority indices of $x$.
    ii) If $x_0  = x_1$, we call the algorithm recursively to determine whether $|x'| \in \inbrace{1,n-3}$ and find a majority index $i$ in $x'$ if $|x'| \in \inbrace{1,n-3}$. Then we discuss the following two cases:
    \begin{enumerate}
        \item  If $|x'| \notin \inbrace{1,n-3}$, we have
        $|x'| \in \inbrace{0,2,...,n-4,n-2}$. Since $x_0 = x_1$ and $|x| = x_0+x_1+|x'|$, we have $|x| \in \inbrace{0,2,...,n-2,n}$. Thus, $|x| \notin \inbrace{1,n-1}$;
        \item  If $|x'| \in \inbrace{1,n-3}$, we compute $x_0 \oplus x'_i$ using 1 quantum query. If $x_0 \neq  x'_i$, then $|x| \in \inbrace{3,n-3}$ and thus $|x| \notin \inbrace{1,n-1}$; if $x_0 = x'_i$, then $|x| \in \inbrace{1,n-1}$.
    \end{enumerate}
 
    \end{itemize}
    Next, we prove the number of queries in the above algorithm to compute $\EXACT_{1,n-1}^n(x)$ is at most $n-2$
    by the induction method. i) For $n=4$ and $5$, the correctness of the proposition is easy to check; ii) We suppose the proposition is correct for any $n$ such that $4 \le n < m$ for some integer $m \ge 6$. Then we aim to prove the correctness of the proposition in the case $n = m$. If $x_0 \neq x_1$, the number of queries in the algorithm is $1+(m-2-1)= m-2$; if $x_0 = x_1$, then by the induction assumption, the number of queries in the algorithm is at most $1+(m-2-2)+1 = m-2$. Thus, the proposition is also correct for $n = m$. %As a result, the query complexity of the above algorithm to compute $\EXACT_{1,n-1}^k(n)$ is $n-2$.

    \end{proof}
Combining Lemma \ref{lemma:klzero}, \ref{lemma:klone} and $Q_E(\EXACT_{k,l}^n) \ge \max\inbrace{n-k,l}-1$ \cite{AIN17}, we prove Theorem \ref{th:kl}, which implies the correctness of Conjecture \ref{con:kl} for the above two cases.

\begin{theorem}\label{th:kl}
If $l-k \ge 2$, then the exact quantum query complexity of $\EXACT_{k,l}^n$ is $\max\inbrace{n-k,l}-1$ for the case $k=0$ and the case $k=1,l=n-1$. 
\end{theorem}

\section{Conclusion}\label{sec:con}

In this paper, we have characterized the exact quantum query complexity of $\Mod_m^n$ for any $1 < m \le n$. As a corollary, we have shown a broad class of symmetric functions is not evasive in the quantum model. Additionally, we have given the tight exact quantum query complexity of $\EXACT_{k,l}^n$ for some cases. Furthermore, there are some open questions worth exploring.

\begin{itemize}
\item For total symmetric Boolean functions, there are still some basic function classes whose quantum exact query complexity has not been fully characterized.  It would be interesting to investigate whether the techniques used in this article can be extended to these functions.

\item How to give a complete characterization to the class of symmetric functions that map from $\B^n$ to a finite set $X$, whose quantum query complexity is less than $n$?

\end{itemize}

The study of the exact quantum query complexity of symmetric functions is an important area of research in quantum computing. While the exact quantum query complexities of a few symmetric functions are well-established, there remain many challenges in this domain. Further research is necessary to enhance our understanding of the exact quantum query complexity of symmetric functions and to explore new quantum algorithms in this field.

\section*{Declaration of competing interest}
The authors declare that they have no known competing financial interests or personal relationships that could have
appeared to influence the work reported in this paper.

\section*{Acknowledgment} 
We would like to thank 
Shabnam Ghalichi, Lvzhou Li, Jingquan Luo and Maris Ozols for pointing out flaws in Algorithm \ref{algorithm1} in early versions of this manuscript. This research was supported by National Natural Science Foundation of China (Grant No. 62332009, 12347104, 61972191) and Innovation Program for
Quantum Science and Technology (Grant No. 2021ZD0302901).

\bibliography{ref}

\end{document}